\documentclass[journal,draftcls,onecolumn,11pt]{IEEEtran}
\usepackage[T1]{fontenc}
\usepackage[latin9]{inputenc}
\usepackage{amsthm}
\usepackage{amsmath}
\usepackage{amssymb}
\usepackage{graphicx}
\usepackage{amsmath,bm}
\usepackage{float}
\usepackage{multirow}
\usepackage{setspace}

\makeatletter

\theoremstyle{plain}
\newtheorem{thm}{\protect\theoremname}
\theoremstyle{definition}
\newtheorem{defn}[thm]{\protect\definitionname}
\theoremstyle{remark}
\newtheorem{rem}[thm]{\protect\remarkname}
\theoremstyle{definition}

\theoremstyle{plain}
\newtheorem{lem}[thm]{\protect\lemmaname}
\theoremstyle{plain}
\newtheorem{prop}[thm]{\protect\propositionname}

\providecommand{\definitionname}{Definition}
\providecommand{\examplename}{Example}
\providecommand{\lemmaname}{Lemma}
\providecommand{\remarkname}{Remark}
\providecommand{\theoremname}{Theorem}
\providecommand{\propositionname}{Proposition}

\makeatother

\begin{document}

\title{HFR Code: A Flexible Replication Scheme for Cloud Storage Systems}

\author{Bing~Zhu, Hui~Li$^{\dagger}$,~\IEEEmembership{Member,~IEEE,} Kenneth~W.~Shum,~\IEEEmembership{Member,~IEEE,} and \\ Shuo-Yen~Robert~Li,~\IEEEmembership{Fellow,~IEEE}

\thanks{A preliminary version of this work was presented at the 2014 National Conference on Communications \cite{key-0}.}

\thanks{B. Zhu is with the School of Electronic and Computer Engineering, Peking University, Shenzhen 518055, P. R. China (e-mail: zhubing@sz.pku.edu.cn).}
\thanks{H. Li is with the Institute of Big Data Technologies and Shenzhen Engineering Lab of Converged Networks Technology, Peking University Shenzhen Graduate School, Shenzhen 518055, P. R. China (e-mail: lih64@pkusz.edu.cn).}
\thanks{K. W. Shum is with the Institute of Network Coding, The Chinese University of Hong Kong, Shatin, Hong Kong (e-mail: wkshum@inc.cuhk.edu.hk).}
\thanks{S.-Y. R. Li is with the University of Electronic Science and Technology of China, Chengdu 611731, P. R. China and Institute of Network Coding, The Chinese University of Hong Kong, Shatin, Hong Kong (e-mail: bobli@ie.cuhk.edu.hk).}
\thanks{$^{\dagger}$Corresponding author.}}

\maketitle

\begin{abstract}
Fractional repetition (FR) codes are a family of repair-efficient storage codes that provide exact and uncoded node repair at the minimum bandwidth regenerating point. The advantageous repair properties are achieved by a tailor-made two-layer encoding scheme which concatenates an outer maximum-distance-separable (MDS) code and an inner repetition code. In this paper, we generalize the application of FR codes and propose heterogeneous fractional repetition (HFR) code, which is adaptable to the scenario where the repetition degrees of coded packets are different. We provide explicit code constructions by utilizing group divisible designs, which allow
the design of HFR codes over a large range of parameters. The constructed codes achieve the system storage capacity under random access repair and have multiple repair alternatives for node failures. Further, we take advantage of the systematic feature of MDS codes and present a novel design framework of HFR codes, in which storage nodes can be wisely partitioned into clusters such that data reconstruction time can be reduced when contacting nodes in the same cluster.
\end{abstract}


\IEEEpeerreviewmaketitle

\begin{spacing}{1.5}

\section{Introduction}

Cloud storage has emerged in recent years as a promising paradigm for storing large amounts of data, and is currently backed by practical cloud storage services such as Amazon S3 and Windows Azure. In a cloud storage system, the original data file is spread across a large number of distinct nodes such that a client can access the data even in the presence of node failures. This distributed storage manner has been widely employed in enterprise systems such as Google File System (GFS)~\cite{key-1}, Hadoop Distributed File System (HDFS)~\cite{key-2}, etc. However, the storage nodes in the cloud system have reached such a massive scale that failures are the norm rather than exception. In order to provide reliable storage over unreliable nodes, redundancy should be introduced to cloud storage systems. The simplest strategy is replication in which the file is available as long as one of the copies remains intact. For example, each data chunk in GFS is replicated several times, with the minimum being three~\cite{key-1}.

For the same level of redundancy, erasure coding techniques can improve data reliability when compared to replication \cite{key-3}. In an erasure coding mechanism, the original data object is encoded and spread over $n$ different nodes in a distributed fashion. A data collector (DC) can reconstruct the source file by attaching to arbitrary $k$ out of the $n$ storage nodes. In other words, the system can tolerate up to $n-k$ concurrent node failures without data loss. Reed-Solomon (RS) code for instance, a classical kind of MDS codes, has been implemented on top of the HDFS for old data archiving \cite{key-4}. Upon failure of a node, the conventional method is to reconstruct the entire source file, re-encode and deliver the encoded block to a replacement node. This process, termed \textit{node repair}, results in large bandwidth overheads since only a small fraction of the downloaded data is stored eventually.

Regenerating code~\cite{key-5} is a class of erasure codes with the capability to minimize the amount of data stored per node (denoted by $\alpha$) and the bandwidth consumption during the repair process. When a node fails, it can be regenerated by contacting any subset of $d$ surviving nodes and downloading $\beta$ packets from each node. These $d$ nodes that participate in the repair process are referred to as \textit{helpers}. For the case that the data stored on the replacement node is identical to that on the failed node, the repair model is called exact repair. Functional repair, in contrast, differs from exact repair that the content of the new node need not be the same as in the failed node, as long as the functionality of the failed node is retained after repair. Furthermore, based on the storage-bandwidth trade-off derived in~\cite{key-5}, two extreme points are obtained by minimizing $\alpha$ and $\beta$ respectively, which correspond to the minimum storage regenerating (MSR) point and minimum bandwidth regenerating (MBR) point. We refer the readers to~\cite{key-6} for more details of regenerating codes (e.g., explicit code constructions and performance evaluations).

Although regenerating codes can reduce the repair overhead as compared to traditional erasure codes, they suffer from the drawback that the repair process is of high complexity that each contacted node has to read all the data it stored in order to transfer a linear combination to the replacement node. Motivated by this, a simplified repair scheme at the MBR point was first introduced in~\cite{key-7} and later extended to a new family of codes with a table-based repair model, termed fractional repetition (FR) codes~\cite{key-8}. Specifically, the repair process of FR codes is characterized by the \textit{uncoded repair} property: a surviving node reads and delivers the exact amount of data to the replacement node without extra arithmetic operations. The desirable repair property is derived from the relaxation in choosing helper nodes: the condition that any $d$ surviving nodes can be contacted for repair (random access repair) is relaxed to specific subsets of $d$ nodes (table-based repair). Clearly, uncoded repair is optimal in terms of repair bandwidth, disk input/output (I/O) and computational complexity, which makes FR codes attractive in practical storage systems.

\subsection{Main Contributions}

In this paper, we focus on a class of erasure codes that provide exact and uncoded repair at the MBR point. Specifically, we propose a flexible replication scheme by relaxing the constraint on repetition degree in conventional FR codes and taking different types of repetition degrees into consideration. We refer to these new codes as \textit{heterogeneous fractional repetition (HFR)} codes. We present explicit constructions of HFR codes from combinatorial designs, in particular from group divisible designs. The designed codes achieve the system storage capacity under the random access repair model and allow to have multiple repair alternatives in the presence of node failures. Furthermore, we present a novel design framework of HFR codes in which the outer MDS code is always made systematic (i.e., one copy of the original data remains in uncoded form). In this framework, we wisely partition the storage nodes into several clusters such that in normal cases, the source file can be read directly from nodes in the same cluster.

\subsection{Organization}

The rest of this paper is organized as follows. Section II reviews the related works on erasure codes that provide uncoded repair. Section III presents the necessary notations and definitions. Section IV introduces the encoding scheme of HFR codes. Section V presents explicit constructions of HFR codes and a novel design framework based on HFR codes. Finally, Section VI concludes this paper with general remarks.

\section{Related Works}

In this section, we review a number of closely related works, such as repair-by-transfer codes~\cite{key-9} and fractional repetition codes~\cite{key-8}.

\subsection{Repair-by-transfer Codes}

In order to simplify the repair process of conventional regenerating codes, the idea of repair-by-transfer is introduced in~\cite{key-9}, where a failed node is repaired by simple transfer of data without the need for any computations. We note that the concept of repair-by-transfer is essentially the same as uncoded repair. The advantageous properties of repair-by-transfer include reduced repair complexity and minimum disk reads at helper nodes. In particular, for the case that $d=n-1$ (i.e., all the surviving nodes participate in the repair process), Shah \textit{et al.} presented in~\cite{key-9} an explicit construction of exact MBR codes with repair-by-transfer via a fully connected graph. In a recent work~\cite{key-10}, a new class of repair-by-transfer codes is proposed based on congruent transformations, in which the minimum field size and encoding complexity of code construction are superior as compared with~\cite{key-9}. Moreover, a different line that focuses on the functional repair model, called functional-repair-by-transfer regenerating code, is presented in~\cite{key-11} and implemented in practical cloud storage systems \cite{key-12}.

\subsection{Fractional Repetition Codes}

Fractional repetition codes are a class of erasure codes characterized by an exact and uncoded repair regime. A tailor-made two-layer encoding process guarantees the desirable repair properties. Specifically, a source file is taken as an input to a pre-determined MDS code, where several coded packets are generated as output. Each coded packet is then replicated $\rho$ times and distributed over distinct storage nodes. The policy that specifies the placement of all coded packets, is called the \textit{repetition code}. As shown in~\cite{key-8}, an FR code with repetition degree $\rho$ can tolerate up to $\rho-1$ concurrent node failures without loss of the exact and uncoded repair property.

We illustrate the construction of FR codes via a simple example. Consider a file composed of $M=6$ packets, we adopt a $(7,6)$ MDS code that outputs $7$ coded packets indexed from $1$ to $7$. Each packet is replicated $\rho=3$ times and spread over $n=7$ distinct nodes, as shown in Figure~\ref{figure1}. Any DC who contacts a subset of $k=3$ nodes can obtain at least 6 distinct packets, which can be used to reconstruct the source file. To illustrate the core idea of exact and uncoded repair, let us assume without loss of generality that node $N_{1}$ is failed, the lost packets can be recovered directly by downloading the corresponding replicas stored on other $d=3$ surviving nodes (e.g., nodes $N_{2}$, $N_{3}$ and $N_{4}$). It is noteworthy that the repairs of FR codes are table-based~\cite{key-8}, which implies that a specific set of nodes (rather than any $d$ nodes) need to be contacted, depending on which node has failed.

\begin{figure}[!t]
\centering
\includegraphics[scale=0.145]{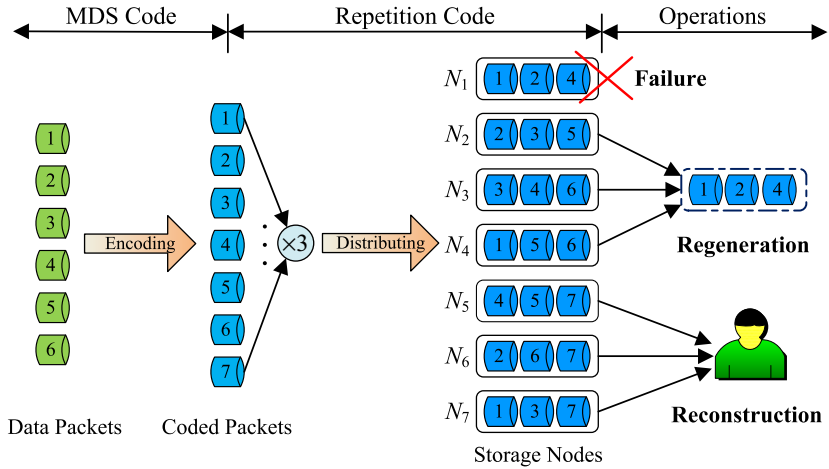}
\caption{A fractional repetition code with repetition degree $\rho=3$.}
\label{figure1}
\end{figure}

Note that MDS codes have been historically investigated and there exist many available constructions in the literature, the main challenge of FR codes lies in the design of the inner repetition code. Existing constructions of FR codes are primarily derived from combinatorial designs. In the pioneer work~\cite{key-8}, an explicit code construction is presented based on regular graphs for the case of single node failure and Steiner systems are adopted in the scenario of multiple failures. In~\cite{key-14}, Koo and Gill considered projective geometries (e.g., bipartite cage graphs and mutually orthogonal Latin squares) in the construction of FR codes. Subsequently, the authors explored in~\cite{key-15} code constructions using resolvable designs, where the packet repetition degree can be varied over a large range. They further considered the usage of Kronecker product to devise new FR codes from existing codes~\cite{key-15-1}. FR codes based on combinatorial configurations are discussed in~\cite{key-16}. Further, optimal code constructions that attain the bounds on the fractional repetition capacity are presented in~\cite{key-17}, which are based on graph theory (e.g., biregular graphs, projective planes and generalized polygons).

Some recent studies propose codes with uncoded repair for heterogeneous systems, where the storage capacities of different nodes may not be the same. Gupta \textit{et al.}~\cite{key-18} introduce weak fractional repetition codes, where each storage node contains different number of packets. Based on the partial regular graph, their construction can tolerate single node failure. Yu \textit{et al.}~\cite{key-19} present irregular fractional repetition codes for heterogeneous storage networks, which can be used to minimize the system repair cost. In our recent work~\cite{key-20}, we propose general fractional repetition codes, which are adaptable to heterogeneous storage capacities and multiple node failures.

While prior works assume that each packet is replicated the same times after the outer MDS code, we consider a different dimension in this paper that the repetition degrees of coded packets can be different.

\section{Preliminaries}

\subsection{Distributed Storage Systems}

Consider a distributed storage system (DSS) with parameters $(n,k,d)$, where $n$ is the total number of nodes in the system, $k$ is the number of nodes that a data collector needs to contact for data reconstruction and $d$ is the repair degree that indicates the number of nodes a replacement node has to contact during the repair process. In normal cases, a system designer should choose the parameters that satisfy
\begin{equation}
k\leq d\leq n-1.
\end{equation}

In this paper, we focus on exact repair where the regenerated packets are identical to the lost packets. During the repair process, a replacement node downloads and stores exactly one packet (i.e., $\beta=1$) each from $d$ helper nodes it contacts, which implies that the storage capacity is $d$ for every node in the system. Under this model, the storage capacity of the DSS, denoted by $C_{MBR}$, is defined as the maximum file size that can be delivered by contacting any $k$ out of the $n$ nodes at the MBR point. In particular, for the case that $\beta=1$, the storage capacity $C_{MBR}$ of an $(n,k,d)$ DSS is shown in~\cite{key-5} to be

\begin{equation}
C_{MBR}=kd-\binom{k}{2}.
\label{capacity}
\end{equation}

It is noteworthy that this expression is derived from the functional repair model. However, Rashmi~\textit{et al.} showed in~\cite{key-7} that MBR codes with exact repair can still achieve the capacity $C_{MBR}$ of (\ref{capacity}).

\subsection{Group Divisible Designs}
\label{GDD}

\begin{defn}
A \textit{combinatorial design} (or simply a \textit{design}) is a pair $(\mathcal{V},\mathcal{B})$, where $\mathcal{V}$ is a set of elements called \textit{points}, and $\mathcal{B}$ is a collection (i.e., multi-set) of nonempty subsets of $\mathcal{V}$ called \textit{blocks}.
\end{defn}

Group divisible design (GDD) is an important academic branch of the design theory. In this subsection, we briefly present the definitions and properties of GDDs, which will be used in our constructions. For an overview of design theory, the reader is referred to~\cite{key-21} and~\cite{key-22}.

\begin{defn}
Let $\Psi$ and $\Theta$ denote sets of positive integers. A group divisible design of order $v$ (denoted by $\Psi$-GDD) is a triple $(\mathcal{V},\mathcal{G},\mathcal{B})$, where $\mathcal{V}$ is a finite set (the \textit{point set}) of cardinality $v$, $\mathcal{G}$ is a partition of $\mathcal{V}$ into parts (\textit{groups}) whose sizes lie in $\Theta$, and $\mathcal{B}$ is a family of subsets (\textit{blocks}) of $\mathcal{V}$ that satisfy the following properties: \end{defn}
\begin{enumerate}
\item $|\mathcal{G}|>1$.
\item If ${B}\in\mathcal{B}$, then $|{B}|\in\Psi$.
\item Every pair of points from distinct groups occurs in a unique block.
\end{enumerate}

If $\Psi=\{\psi\}$, then $\Psi$-GDD is abbreviated as $\psi$-GDD. If $v=u_{1}\theta_{1}+u_{2}\theta_{2}+\ldots+u_{s}\theta_{s}$ and if there are $u_{i}$ groups of size $\theta_{i}, i=1,2,\ldots,s$, then the $\Psi$-GDD is said to be of type $\theta_{1}^{u_{1}}\theta_{2}^{u_{2}}\cdots \theta_{s}^{u_{s}}$. For example, if $\mathcal{V}=\{1,\ldots,5\}$, $\mathcal{G}=\{\{1\},\{2,3\},\{4,5\}\}$ and $\mathcal{B}=\{\{2,5\},$ $\{3,4\},\{1,2,4\},\{1,3,5\}\}$, then $(\mathcal{V},\mathcal{G},\mathcal{B})$ is a $\{2,3\}$-GDD of type $1^{1}2^{2}$. Throughout this paper, we will use the exponential notation of GDDs.

\begin{rem}
Let $\mathcal{D}=(\mathcal{V},\mathcal{B})$ be a design with points $\mathcal{V}=\{p_{1},\ldots,p_{v}\}$ and blocks $\mathcal{B}=\{B_{1},\ldots,B_{b}\}$. The \textit{incidence matrix} $\mathbf{A}=(a_{ij})$ of $\mathcal{D}$ is a $v\times b$ matrix defined by $a_{ij}=1$ if $p_{i}$ belongs to $B_{j}$ and $a_{ij}=0$ otherwise. For example, the incidence matrix of the above $\{2,3\}$-GDD is

\[
\mathbf{A}=\left[\begin{array}{cccc}
0 & 0 & 1 & 1\\
1 & 0 & 1 & 0\\
0 & 1 & 0 & 1\\
0 & 1 & 1 & 0\\
1 & 0 & 0 & 1
\end{array}\right].
\]
\end{rem}

Note that if all the blocks in a GDD are of the same size, then the points from the same group belong to the same number of blocks, whereas the points from different groups may occur in different number of blocks. Specifically, given a $\psi$-GDD of type $\theta_{1}^{u_{1}}\theta_{2}^{u_{2}}\cdots \theta_{s}^{u_{s}}$, a point from the group of size $\theta_{i}$  is contained in $r_{i}$ blocks, where

\begin{equation}
r_{i}=\frac{[\sum_{j=1}^{s}u_{j}\theta_{j}]-\theta_{i}}{\psi-1}, i=1,2,\ldots,s.
\label{frequency}
\end{equation}

\section{Heterogeneous Fractional Repetition Codes}

In practical cloud storage systems, it is not necessary to have all the data chunks with the same repetition degree. Due to some special considerations, a fraction of the stored data chunks can be replicated more than others. We present here two observations:

\begin{itemize}
\item \textbf{Popularity:} For a data object that consists of a large number of chunks, the popularity of each data chunk can be different. From a system perspective, some chunks may be frequently accessed, while others are rarely accessed by users. Therefore, having the popular chunks with more replicas allows for parallel accesses.
\end{itemize}

\begin{itemize}
\item \textbf{Functionality:} In a practical deployment setting, the functionalities of different packets may not be the same. It is desirable to have more important packets with higher repetition degrees. For example, the systematic packets of an MDS code deserve higher repetition degrees, since they provide direct data reconstruction without decoding.
\end{itemize}

Heading towards this direction, we relax the constraint on the repetition degree in prior works and introduce \textit{heterogeneous fractional repetition (HFR)} codes, where the repetition degrees of different packets can be different. Note that in analogy with FR codes, the encoding process of HFR codes is formed by the concatenation of two components: an outer MDS code followed by an inner heterogeneous repetition code. However, the repetition degrees of coded packets, which are determined by the heterogeneous repetition code, are not the same.

In the encoding process of HFR codes, we first employ an outer $(\phi,M)$ MDS code to encode a source file consisting of $M$ packets, where $\phi$ coded packets are generated with indices $\{1,\ldots,\phi\}$ respectively. Further, we partition these $\phi$ coded packets into $\lambda$ subsets $P_{1},P_{2},\ldots,P_{\lambda}$ such that each packet from $P_{i}$ is replicated $\rho_{i}$ times, where $\rho_{i}\neq\rho_{j},\forall i,j\in\{1,\ldots,\lambda\},i\neq j$. After the replication operation, all coded packets and the replicas are distributed over a large number of distinct storage nodes, which may spread over a wide geographical area. The policy that specifies the placement of all the coded packets, is called the \textit{heterogeneous repetition code}.

We refer to a HFR code as a $\lambda$-HFR code if there exist $\lambda$ types of repetition degrees. It is clear that if $\lambda=1$, then the $\lambda$-HFR code reduces to an FR code.

\subsection{Definition}

We present now the formal definition of HFR codes. For a positive integer $n$, we use $[n]$ to denote the set $\{1,2,\ldots,$ $n\}$.

\begin{defn}
\label{Def_HFR}
Let $\lambda$ and $\phi$ be positive integers, $\rho_{1},\ldots,\rho_{\lambda}$ be distinct positive integers, and $\Phi_{1},\ldots,\Phi_{\lambda}$ be a partition of $[\phi]$. A $\lambda$-HFR code $\mathcal{C}$ for an $(n,k,d)$ DSS, is a collection of $n$ subsets $\xi_{1},\ldots,\xi_{n}$ of $[\phi]$ that satisfies the following conditions:
\end{defn}
\begin{enumerate}
\item The cardinality of each $\xi_{i}$ is $d$, i.e., $|\xi_{i}|=d, 1\leq i\leq n$.
\item Each point in $\Phi_{i}$ occurs in exactly $\rho_{i}$ subsets in $\mathcal{C}$, where $1\leq i\leq \lambda$.
\item Any pair of distinct points of $[\phi]$ is contained in at most one subset, i.e., $\left|\xi_{i}\cap\xi_{j}\right|\leq1$, $\forall i,j\in[n]$. 
\end{enumerate}

In the above definition, each subset corresponds to a specific storage node and the points in the subsets represent the indices of coded packets that are stored on the node. In other words, the node capacity is given by the cardinality of the subset. Furthermore, we note that $\Phi_{i}$ represents the packet set $P_{i}$, i.e., the points in $\Phi_{i}$ are the indices of coded packets in $P_{i}, 1\leq i\leq\lambda$. In this sense, we obtain that
\begin{equation}
\sum_{i=1}^{n}|\xi_{i}|=\overset{\lambda}{\underset{j=1}{\sum}}\rho_{j}|\Phi_{j}|.
\end{equation}

We use $S_{\mathcal{C}}(k)$ to denote the number of guaranteed distinct packets of an HFR code $\mathcal{C}$ when contacting any $k$ nodes, i.e.,
\begin{equation}
S_{\mathcal{C}}(k)\triangleq\underset{K\subset[n],|K|=k}{\min}\bigl|\cup_{i\in K}\xi_{i}\bigr|.
\end{equation}

For a DSS with parameters $(n,k,d)$, $S_{\mathcal{C}}(k)$ denotes the maximum file size that can be stored using $\mathcal{C}$. By following the same approach as~\cite[Lemma 14]{key-8}, we can obtain the following upper bound on $S_{\mathcal{C}}(k)$,
\begin{equation}
\label{Bound}
S_\mathcal{C}(k)\leq\Big\lfloor\sum_{j=1}^\lambda|\Phi_j|\Big(1-\frac{\binom{n-\rho_j}{k}}{\binom{n}{k}}\Big)\Big\rfloor.
\end{equation}
The proof of (\ref{Bound}) is omitted. We refer to an HFR code as \textit{optimal} if it attains the upper bound.

\subsection{Incidence Matrix Viewpoint}

In this subsection, we provide an alternative viewpoint based on incidence matrix which completely characterizes the properties of HFR codes. Consider a $\lambda$-HFR code $\mathcal{C}=\{\xi_{1},\ldots,\xi_{n}\}$ of a symbol set $\{s_{1},\ldots,s_{\phi}\}$, we can construct a \textit{symbol-by-subset} incidence matrix $\mathbf{M}$ with its entries $m_{ij}$ given as
\[
m_{ij}=\begin{cases}
1, & s_{i}\in\xi_{j};\\
0, & otherwise.
\end{cases}
\]

Let $\bm{r}_{i},\bm{c}_{j}$ denote the rows and columns of $\mathbf{M}$ respectively, where $i\in[\phi],j\in[n]$. Note that the \textit{support} of a nonzero vector $\bm{o}=(o_{1},\ldots,o_{z})$, is the set of indices of its nonzero coordinates: $supp(\bm{o})=\{x|o_{x}\neq0\}$. In this sense, the first condition of HFR codes then translates to
\begin{equation}
\label{HFR_existence1}
|supp(\bm{c}_{j})|=d, 1\leq j\leq n.
\end{equation}

Assume that $\{T_{1},T_{2},\ldots,T_{\lambda}\}$ is a partition of $\{s_{1},\ldots,s_{\phi}\}$, where $T_{k}=\{s_{i}|i\in\Phi_{k}\},k\in[\lambda]$. Thus, the second condition in Definition~\ref{Def_HFR} is equivalent to that if $s_{i}\in T_{k}$,
\begin{equation}
|supp(\bm{r}_{i})|=\rho_{k}, 1\leq i\leq\phi.
\end{equation}

We use $\Lambda_{i}$ to denote the number of subsets containing $s_{i}$ and $\Lambda_{i,i'}$ to denote the number of subsets containing $s_{i},s_{i'}$ simultaneously, where $i,i'\in[\phi]$. 
Let $\mathbf{M}^{T}$ denote the transpose of $\mathbf{M}$. We thus have

\[
\mathbf{M}\mathbf{M}^{T}=\left[\begin{array}{cccc}
\Lambda_{1} & \Lambda_{1,2} & \cdots & \Lambda_{1,\phi}\\
\Lambda_{2,1} & \Lambda_{2} & \cdots & \Lambda_{2,\phi}\\
\vdots & \vdots & \ddots & \vdots\\
\Lambda_{\phi,1} & \Lambda_{\phi,2} & \cdots & \Lambda_{\phi}
\end{array}\right].
\]

Therefore, the third condition of HFR codes can be interpreted as
\begin{equation}
\label{HFR_existence2}
\Lambda_{i,i'}\leq1,\forall i,i'\in[\phi].
\end{equation}

We use the incidence matrix viewpoint to prove the necessary and sufficient conditions for HFR codes. Specifically, consider a design $\mathcal{D}$ with point set $\mathcal{V}$ and block set $\mathcal{B}$, we can construct a $\lambda$-HFR code by setting $[\phi]=\mathcal{V}$ and $\mathcal{C}=\mathcal{B}$, if the following theorem suffices.

\begin{thm}
The blocks in a design $\mathcal{D}=(\mathcal{V},\mathcal{B})$ form a $\lambda$-HFR code if and only if the incidence matrix of $\mathcal{D}$ satisfies conditions (\ref{HFR_existence1})\textendash{}(\ref{HFR_existence2}).
\end{thm}

\begin{IEEEproof}
Necessity: If there exists a $\lambda$-HFR code $\mathcal{C}$, the symbol-by-subset incidence matrix $\mathbf{M}$ of $\mathcal{C}$ should satisfy conditions (\ref{HFR_existence1})\textendash{}(\ref{HFR_existence2}). By setting $[\phi]=\mathcal{V}$ and $\mathcal{C}=\mathcal{B}$, then $\mathbf{M}$ corresponds to the incidence matrix $\mathbf{A}$ of the design $\mathcal{D}$. Thus, $\mathbf{A}$ subjects to conditions (\ref{HFR_existence1})\textendash{}(\ref{HFR_existence2}).

Sufficiency: If the incidence matrix of adopted design subjects to (\ref{HFR_existence1})\textendash{}(\ref{HFR_existence2}), the three conditions of $\lambda$-HFR codes will be satisfied from the construction. Hence, the designed code is a $\lambda$-HFR code by definition.
\end{IEEEproof}

\section{HFR Codes from Combinatorial Designs}

In this section, we first provide an explicit construction of $\lambda$-HFR codes from group divisible designs. We further propose a novel design framework based on
HFR codes.

\subsection{Code Constructions}

We start by presenting an example before introducing the code construction. Consider a 3-GDD of type $1^{4}3^{1}$, where the groups are partitioned as $\mathcal{G}=\{\{1\},\{2\},\{3\},\{4\},\{5,6,7\}\}$. Correspondingly, the block set is $\mathcal{B}=\{\{1,2,5\},$ $\{1,3,6\},\{1,4,7\},\{2,3,7\},\{2,4,6\},\{3,4,5\}\}$. We adopt a $(7,4)$ systematic MDS code to encode a data file of $M=4$ packets, where the systematic packets are output with indices $\{1,\ldots,4\}$ and parities are with indices $\{5,6,7\}$ respectively. Each systematic packet is then repeated $3$ times, while the parity packets are only repeated twice. As shown in Figure~\ref{Example_HFR}, if we take $\mathcal{C}=\mathcal{B}$, we can place these coded packets on $n=6$ nodes according to the $2$-HFR code $\mathcal{C}$.
Furthermore, we observe that any DC contacting $k=3$ nodes can obtain at least $6$ distinct packets. By using (\ref{Bound}), we can compute that $S_\mathcal{C}(3)\leq\lfloor6.2\rfloor=6$, which indicates that the $2$-HFR code $\mathcal{C}$ is optimal.

\begin{figure}
\centering{}\includegraphics[scale=0.172]{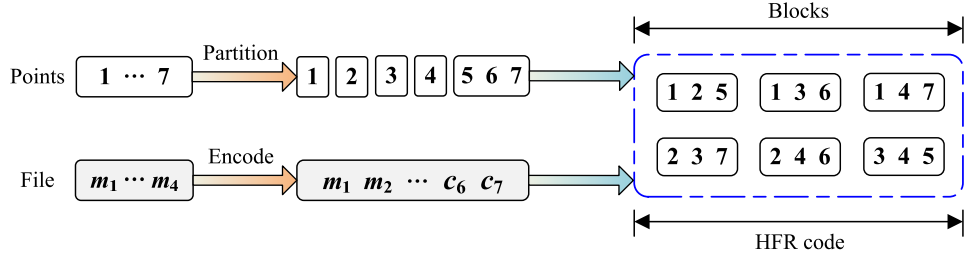}
\caption{A $2$-HFR code from 3-GDD of type $1^{4}3^{1}$.}
\label{Example_HFR}
\end{figure}

We present now explicit constructions of HFR codes from group divisible designs. Note that the nodes in the system, which correspond to blocks, are of the same storage capacity. We thus focus on GDDs with the same block size. The core idea behind the construction is that for a $\psi$-GDD of type $\theta_{1}^{u_{1}}\theta_{2}^{u_{2}}\cdots \theta_{s}^{u_{s}}$, the points from different groups occur in different number of blocks. If we view the points as the indices of coded packets and the blocks as storage nodes, we can obtain different repetition degrees.

\textbf{Construction:} Given a $\psi$-GDD of type $\theta_{1}^{u_{1}}\cdots \theta_{\lambda}^{u_{\lambda}}$ with block set $\mathcal{B}$, a $\lambda$-HFR code $\mathcal{C}$ can be obtained by taking $\phi=u_{1}\theta_{1}+\ldots+u_{\lambda}\theta_{\lambda}$ and $\mathcal{C}=\mathcal{B}$. Correspondingly, the $\lambda$ repetition degrees are given by
\begin{equation}
\rho_{l}=\frac{[\sum_{w=1}^{\lambda}u_{w}\theta_{w}]-\theta_{l}}{\psi-1},l=1,\ldots,\lambda.
\end{equation}

The constructed $\lambda$-HFR code can be implemented in a DSS with parameters $n=\sum_{i=1}^{\lambda}u_{i}\theta_{i}\rho_{i}/\psi$ and $d=\psi$. Note that our construction is based on the assumption that there exists a GDD of certain types, of which the existence results are presented in \cite[Part IV]{key-22}. For instance, Abel \textit{et al.} proved in \cite{key-23} the necessary conditions for the existence of a $5$-GDD of type $g^5m^1$:

\begin{lem}
A $5$-GDD of type $g^5m^1$ exists if $g\equiv 0\pmod 4$, $m\equiv 0\pmod 4$ and $m\leq 4g/3$, with the possible exceptions of $(g,m)=(12,4)$ and $(12,8)$.
\end{lem}

Further, since any two nodes intersect at not more than one point, we thus obtain the following theorem.
\begin{thm}
The $\lambda$-HFR codes designed by the Construction above achieve the system storage capacity $C_{MBR}$ under the random access repair model.
\end{thm}

\begin{proof}
 From the properties of group divisible designs, we have that any pair of points from different groups is contained in exactly one block. In other words, any two distinct nodes can not share more than one packet. Therefore, there are at most $\tbinom{k}{2}$ repeated packets when contacting $k$ nodes of the system, which suggests that HFR codes obtained by the Construction achieve the storage capacity $C_{MBR}$.
\end{proof}

\begin{rem}
The metric of repair alternativity is introduced in~\cite{key-24}, which measures the number of different subsets of nodes that enable the recovery of a failed node. We note that an erasure code with large repair alternatives increases the probability to recover lost data even in the presence of multiple failures.
\end{rem}

Due to the heterogeneity in repetition degrees, the repair alternativity of HFR codes depends on the packets that stored on the failed node. Without loss of generality, we assume that the repetition degrees of the lost packets are $\rho_{1},\ldots,\rho_{d}$ respectively. Since any pair of distinct points is contained in exactly one group or one block, the remaining replicas of all lost packets are spread over different nodes such that there will not exist a surviving node that contains one pair of the remaining replicas. Therefore, the repair alternativity of the constructed $\lambda$-HFR code is $\prod_{i=1}^{d}(\rho_{i}-1)$.

For an $(n,k,d)$ DSS that implements HFR codes, the repair alternativity of each node can reach a high level as node storage capacity and packet repetition degree increase.

\subsection{A Novel Design Framework}

In practical storage systems, it is desirable to have the systematic feature of MDS codes, for that the original file can be read directly from the uncoded copy without decoding. In this sense, we can always adopt systematic MDS code as the outer code. Recall that a DC connecting to any $k$ nodes can recover the source file. However, the reconstruction overheads can vary over a wide range, depending on which specific set of $k$ nodes have been contacted. We illustrate this finding by the following example.

Consider a systematic $(10,6)$ MDS code with $6$ systematic packets of indices $\{1,\ldots,6\}$ and $4$ parity packets of indices $\{7,\ldots,10\}$ respectively. Each systematic packet is repeated $4$ times, while each parity packet is only replicated $3$ times. We distribute these coded packets over $n=12$ nodes according to the $2$-HFR code:
\[
\{1,3,7\},\{1,4,8\},\{1,5,9\},\{1,6,10\},\{2,6,7\},\{2,5,8\},
\]
\[
\{2,3,9\},\{2,4,10\},\{4,5,7\},\{3,6,8\},\{4,6,9\},\{3,5,10\}.
\]

We label the nodes by $N_{1},\ldots,N_{12}$ respectively, from left to right and top to bottom. It can be observed that any DC who connects to $k=3$ nodes can obtain at least 6 distinct packets, which are adequate to retrieve the source file. However, if he connects to the $3$ nodes $N_{1},N_{5}$ and $N_{9}$, then he will get all the $6$ systematic packets and no additional operations are needed. While contacting nodes $N_{1},N_{3}$ and $N_{12}$, he will only obtain 3 systematic packets. Therefore, it is necessary to perform some decoding operations on the obtained parity packets, such as matrix inversion and linear combinations.

\begin{table}
\caption{Simulation Results}
\centering{}%
\begin{tabular}{c|c|c|c}
\hline
$N(s)$ & $N(p)$ & Occurrence Number & Frequency\tabularnewline
\hline
\hline
$3$ & $3$ & $0.363\times10^{7}$ & $0.0363\%$\tabularnewline
\hline
\multirow{2}{*}{$4$} & $2$ & $1.636\times10^{7}$ & \multirow{2}{*}{$0.4909\%$}\tabularnewline
\cline{2-3}
 & $3$ & $3.273\times10^{7}$ & \tabularnewline
\hline
\multirow{2}{*}{$5$} & $2$ & $3.273\times10^{7}$ & \multirow{2}{*}{$0.4364\%$}\tabularnewline
\cline{2-3}
 & $3$ & $1.091\times10^{7}$ & \tabularnewline
\hline
\multirow{2}{*}{$6$} & $1$ & $0.182\times10^{7}$ & \multirow{2}{*}{$0.0364\%$}\tabularnewline
\cline{2-3}
 & $3$ & $0.182\times10^{7}$ & \tabularnewline
\hline
\end{tabular}
\label{simulation}
\end{table}

\textbf{Simulation:} We simulate the behavior of a DC and pick a set of $k=3$ distinct nodes at random from the $12$ nodes. Among all the distinct packets obtained from the contacted nodes, we use $N(s),N(p)$ to denote the number of systematic packets and parity packets respectively. In the simulation, we randomly contact $3$ nodes $10^{8}$ times. From Table~\ref{simulation}, we observe that a DC can obtain all the $6$ systematic packets (regardless of the number of parity packets) with less than $4\%$ probability. In other words, the DC needs to complete the decoding operations with more than $96\%$ probability.

\begin{figure}
\begin{centering}
\includegraphics[scale=0.15]{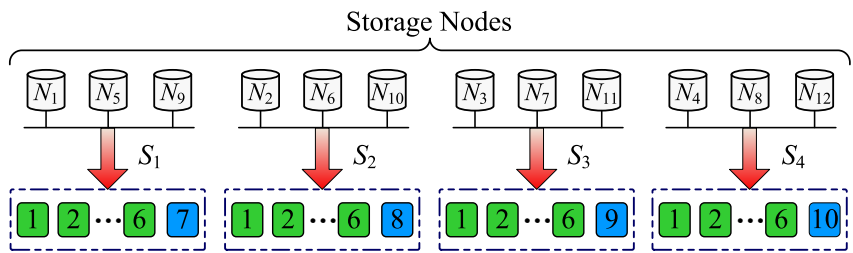}
\par\end{centering}
\caption{All the storage nodes are partitioned into $4$ virtual clusters.}
\label{cluster}
\end{figure}

The simulation results indicate the inefficiency of data reconstruction under the random access model adopted in the literature. Motivated by this, we carefully partition the above $12$ nodes into $4$ virtual clusters $S_{1},\ldots,S_{4}$, where each cluster contains one copy of the systematic packets, as depicted in Figure~\ref{cluster}. In this sense, a data collector can readily reconstruct the source file without decoding if he connects to any one of the four clusters.

It is clear that the cluster-based nodes partition model can reduce the data reconstruction time as compared to the random access model. We note that the partition may allude to some practical implications. For instance, the source file may be mainly needed by clients from several different geographical regions in practical systems. If we deploy the nodes from each cluster into each region, the clients can read the file in a simple manner. Moreover, since practical deployments always contain a track server storing the system metadata, we can write the cluster details as well as the repair table~\cite{key-8} into the metadata, which can be read directly to facilitate both data reconstruction and node recovery.

In general, the $n$ nodes of an HFR code can be partitioned into $k$-sized clusters if each cluster contains all the $M$ original packets. Suppose that $n=tk,t\geq2$ and the data packets are indexed by $\{1,\ldots,M\}$.

\begin{prop}
Given a $\lambda$-HFR code $\mathcal{C}=\{\xi_{1},\ldots,\xi_{n}\}$, let $H_{1},H_{2},\ldots,H_{t}$ be a partition of $\mathcal{C}$, where $H_{i}=\{\xi_{i_1},\ldots,\xi_{i_k}\},1\leq i\leq t$. Let $D_{i}$ denote the collection of distinct points in $H_{i}$, i.e., $D_{i}=\bigcup_{\xi\in H_{i}}\xi$. If $[M]\subset D_{i}$, the $k$ nodes in each $H_{i}$ can be viewed as a cluster.
\end{prop}

We present now the necessary conditions that an HFR code constructed from GDDs can be partitioned with the desirable reconstruction property. Note that there are two types of packets in a systematic MDS code, we thus focus on the case that $\lambda=2$. We further assume that each node contains exactly one parity packet and nodes with the same parity packet are included in the same cluster.

Consider a systematic MDS code with parameters $(a\theta_{1}+\theta_{2},a\theta_{1})$, where the systematic symbols are labeled by $\{1,\ldots,a\theta_{1}\}$ and parity packets by $\{a\theta_{1}+1,\ldots,a\theta_{1}+\theta_{2}\}$ respectively. In our construction, we adopt a class of $\psi$-GDDs of type $\theta_{1}^{a}\theta_{2}^{1}$. The points from the groups of size $\theta_{1}$ give the indices of the systematic symbols and points in the group of size $\theta_{2}$ will represent parities. We assume that systematic symbols are of the higher repetition degrees, which implies that $\theta_{1}<\theta_{2}$. Since each cluster contains one copy of the systematic packets and one distinct parity packet, it is clear that the repetition degree of each systematic packet is equal to the number of parity packets, i.e.,
\begin{equation}
\theta_{2}=\frac{(a-1)\theta_{1}+\theta_{2}}{\psi-1}.
\end{equation}

The above $2$-HFR codes are constructed assuming the existence of $\psi$-GDDs of desired parameters. We note that the necessary conditions for $\psi$-GDDs of type $\theta_{1}^{a}\theta_{2}^{1}$ for small block sizes (i.e., $\psi=3,4,5$), are given in~\cite{key-22}. If we evaluate all the possible group types, then we have the appropriate designs of small order, such as 3-GDD of type $\{1^{4}3^{1},2^{3}4^{1},1^{6}5^{1}\}$ and 4-GDD of type $1^{9}4^{1}$, etc.

\section{Conclusion}

In this paper, we generalize the application of FR codes and propose HFR codes, in which the repetition degrees of coded packets are different. We present explicit constructions of HFR codes derived from group divisible designs. We show that the constructed HFR codes achieve the system storage capacity under random access repair and have multiple repair alternatives upon node failure. Further, we present a novel design framework based on HFR codes where storage nodes can be wisely partitioned into several clusters. The proposed framework in conjunction with the systematic feature of MDS codes can significantly reduce the data reconstruction time. Exploring constructions of HFR codes from other combinatorial structures is an interesting direction for future research.

\section*{Acknowledgment}

This work is supported in part by the National Basic Research 973 Program of China under Grant No. 2012CB315904, National Natural Science Foundation of China under Grant No. 61179028, Natural Science Foundation of Guangdong under Grant S2013020012822, Shenzhen Basic Research Project under Grant JCYJ20130331144502026, and a grant from the University Grants Committee of the Hong Kong Special Administrative Region, China (Project No. AoE/E-02/08).

\end{spacing}

\end{document}